\documentclass[11pt,a4paper]{article}
\usepackage{geometry,color}

\usepackage[latin1]{inputenc}

\usepackage{amsmath,amsthm,amssymb,subdepth}
\usepackage{amsfonts,mathrsfs,enumerate}
\usepackage{setspace}
\usepackage{hyperref}
\usepackage{url}
\title{The Hilbert manifold of asymptotically flat metric extensions}
\date{\today}
\author{Stephen\hspace{-1.5mm}\newlength{\Mheight}
\newlength{\cwidth}
\settoheight{\Mheight}{M}\settowidth{\cwidth}{c}M\parbox[b][\Mheight][t]{\cwidth}{c}\hspace{0.15mm}Cormick\footnote{stephen.mccormick@une.edu.au}\vspace{3mm}\\School of Science and Technology\\University of New England\\Armidale, 2351\\Australia}
\newtheorem{theorem}{Theorem}[section]

\newtheorem{corollary}[theorem]{Corollary}

\newtheorem{definition}[theorem]{Definition}

\newtheorem{lemma}[theorem]{Lemma}

\newtheorem{prop}[theorem]{Proposition}
\newtheorem{remark}[theorem]{Remark}

\newcommand{\onabla}{\mathring\nabla}

\newcommand{\og}{\mathring g}

\newcommand{\xiref}{\xi_{\tref}}

\newcommand{\R}{\mathbb{R}}

\newcommand{\bH}{\overline{H}}

\newcommand{\Hdir}{H^2_\delta\cap \bH^1_\delta}
\newcommand{\Hdirh}{H^2_{-1/2}\cap \bH^1_{-1/2}}

\numberwithin{equation}{section}

\DeclareMathOperator{\tr}{tr}

\DeclareMathOperator{\tref}{ref}

\usepackage{fancyhdr}
\begin{document}
\maketitle
\begin{abstract}
In [Comm. Anal. Geom., 13(5):845-885, 2005.], Bartnik described the phase space for the Einstein equations, modelled on weighted Sobolev spaces with local regularity $(g,\pi)\in H^2\times H^1$. In particular, it was established that the space of solutions to the contraints form a Hilbert submanifold of this phase space. The motivation for this work was to study the quasilocal mass functional now bearing his name. However, the phase space considered there was over a manifold without boundary. Here we demonstrate that analogous results hold in the case where the manifold has an interior compact boundary, where Dirichlet boundary conditions are imposed on the metric. Then, still following Bartnik's work, we demonstrate the critical points of the mass functional over this space of extensions correspond to stationary solutions. Furthermore, if this solution is sufficiently regular then it is in fact a static black hole solution. In particular, in the vacuum case, critical points only occur at exterior Schwarzschild solutions; that is, critical points of the mass over this space do not exist generically. Finally, we briefly discuss the case when the boundary data is Bartnik's geometric data.

\end{abstract}
\section{Introduction}
It is well-known that the total mass/energy of an isolated body in general relativity is given by the ADM mass, and that the very nature of general relativity precludes the possibility of a local energy density; however, the notion of the mass contained in a given region of finite extent is still an open problem. This question is particularly peculiar, as it is not that we lack an answer to it, but rather we have many candidates for what this mass should be (See \cite{QLMreview} for a detailed review), many of which are incompatible. Bartnik's quasilocal mass \cite{qlm} is considered by many to give the best answer to this question, if only it were possible to compute in general. The Bartnik mass is described as follows: Given a subset $\Omega$ of some $(\tilde M,\tilde g,\tilde\pi)$, an initial data set satisfying the Einstein constraints, let $\mathcal{PM}$ be the set of asymptotically flat initial data sets satisfying the positive mass theorem, in which $\Omega$ isometrically embeds, with no horizons outside of $\Omega$. The Bartnik mass is then taken as the infimum of the ADM mass over $\mathcal{PM}$. It is conjectured that this infimum is indeed realised; however, while some progress has been made (see \cite{AndersonKhuri,Bartnik02,Corvino2000,MantoulidisSchoen,pengzi2003} and references therein), this is still an open problem in general. In the case where $\Omega$ is bounded by a minimal surface, this is known to be false; in a recent paper of Mantoulidis and Schoen \cite{MantoulidisSchoen}, a sequence of extensions to a stable minimal surface is constructed, whose mass converges to the Bartnik mass. In light of black hole uniqueness theorems, the only possible limit for this sequence is a Schwarzschild solution, so if $\Omega$ is not contained in a slice of Schwarzschild then the infimum is not realised. However, this conjecture is still wide open when the boundary is not a minimal surface.

There are also interesting results by Corvino \cite{Corvino2000} and Miao \cite{pengzi2003}, which demonstrate that if a mass-minimising extension exists, then it must be static and satisfy Bartnik's geometric boundary conditions \cite{tsinghua}. That is, the metric is Lipshitz across the boundary and the mean curvature on each side of the boundary agree. Bartnik's work on the phase space for the Einstein equations \cite{phasespace} was, in part, motivated by the idea of putting Corvino and Miao's work in a more variational setting. Here we work to this end. For more details pertaining to the space $\mathcal{PM}$ and the Bartnik mass, the reader is referred to \cite{qlm,tsinghua}. In this paper, we consider a larger set of extensions to such a bound domain $\Omega$, described by asymptotically flat manifolds with Dirichlet boundary conditions imposed on a compact interior boundary, $\Sigma$. The initial data we consider has local regularity $(g,\pi)\in H^2\times H^1$, with $g$ prescribed on $\Sigma$ in the trace sense.

The structure of this article is as follows. In Section \ref{Slaplace}, we review the mapping properties of the Laplace-Beltrami operator $M$ and show that this is an isomorphism between certain weighted spaces over $M$ with Dirichlet boundary conditions imposed. In Section \ref{Sphasespace}, we apply Bartnik's phase space analysis to the case considered here, where $M$ has an interior boundary and $g$ satisfies Dirichlet boundary conditions. In particular, we prove that the space of asymptotically flat solutions to the constraints, satisfying the Dirichlet boundary conditions, is a Hilbert manifold. Finally, in Section \ref{Scriticalpoints}, we prove a result intimately related to the static metric extension conjecture and Bartnik's quasilocal mass. We prove that critical points of the mass over the space of extensions to $\Omega$ with Dirichlet boundary conditions, correspond to stationary solutions with vanishing stationary Killing vector on $\Sigma$. In particular, if $g$ is sufficiently smooth, this implies $\Sigma$ is the bifurcation surface of a bifurcate Killing horizon that is non-rotating, and by a staticity result of Sudarsky and Wald \cite{SW1}, one expects that the extension is therefore static. We also obtain a version of this result related to the geometric boundary data (Corollary \ref{cormcrit2}).

\section{The Laplace-Beltrami operator on an asymptotically flat manifold with interior boundary}\label{Slaplace}
The constraint equations form a system of geometric PDEs that do not conform exactly to any of the standard classifications; however, it is well-known that morally they behave as an elliptic system. In fact, a great deal of the research on the constraint equations explicitly relies on this ``morally elliptic" structure. For this reason, we first discuss some preliminary results regarding the Laplace-Beltrami operator on an asymptotically flat manifold with interior boundary. The results in this section are to be entirely expected in light of classical results and their counterparts on asympotically flat manifolds without boundary (cf. \cite{AF,ellipticsys,McOwen}), however, it worthwhile to present them here.

It is well-known that while the Laplace operator is not Fredholm on $\R^n$ when considered as a map $H^2\rightarrow L^2$, it is in fact an isomorphism between certain weighted Sobolev/Lebesgue spaces (cf. \cite{NW}). In this section we discuss some properties of the Laplace-Beltrami operator on an asymptotically flat manifold when Dirichlet boundary conditions are imposed.

Throughout, we let $M$ be a smooth, connected manifold with compact boundary $\Sigma$. Further assume that there exists a compact set $K\subset M\cup\Sigma$ such that the complement $M\cup\Sigma\setminus K$ consists of $N$ connected components, each diffeomorphic to $\R^n$ minus the closed unit ball, $B$. For concreteness, we denote these connected components by $N_i$, and the associated diffeomorphisms by $\phi_i:N_i\rightarrow \R^n\setminus B$. Equip $M$ with a smooth background Riemannian metric $\og$, equal to the pullback of the Euclidean metric to each of these ends. Let $r$ be a smooth function on $M$ such that $r(x)=|\phi_i(x)|$ on each $N_i$, and $\frac12<r<2$ on $K$.

In terms of this background asymptotically flat structure, we define the usual weighted Lebesgue and Sobolev norms, respectively as follows:

\begin{align}
\left\|u\right\|_{p,\delta}&=
\left\{
\begin{array}{ll}
\left(\int\left| u\right|^p r^{-\delta p-n}d\mu_0\right)^{1/p},& p<\infty,\\
\text{ess sup}(r^{-\delta}|u|), & p=\infty,
\end{array}\label{weighted1}
\right.
\\
\left\|u\right\|_{k,p,\delta}&=\sum_{j=0}^k\|\onabla^j u\|_{p,\delta-j}.\label{weighted2}
\end{align}
Norms of sections of bundles are defined in the usual way. Note that our convention follows \cite{AF}, where $\delta$ directly indicates the asymptotic behaviour; that is, $u\in L^p_\delta$ behaves as $o(r^\delta)$ near infinity. We denote the completion of the smooth compactly supported functions with respect to these norms, by $L^p$ and $\overline{W}^{k,p}_\delta$. Note that $\overline{W}^{k,p}_\delta$ is a space of functions that vanish on the boundary in the trace sense, along with their first $k-1$ derivatives. We use $W^{k,p}_\delta$ to denote the completion of the smooth functions with bound support, and also use the convention $\bH^k_\delta=\overline{W}^{k,2}_\delta$ and $H^2_\delta=W^{k,2}_\delta$.

It is well-known that weighted versions of the usual Sobolev-type inequalities hold for these norms; see, for example, Theorem 1.2 of \cite{AF}. While these inequalities are generally considered on manifolds without boundary, it is obvious that the proofs remain valid when a boundary is present. One can easily check this, as the proof in \cite{AF} relies only on splitting the norms into integrals over annular regions, rescaling the integrals to integrals over an annulus of fixed radius, then applying the usual local inequalities. The reader is referred to \cite{ellipticsys,McOwen} for more results pertaining to these weighted spaces.

In terms of these weighted Sobolev spaces, we make precise the notion of asymptotically flat manifolds considered here.
\begin{definition}\label{defAF}
An asymptotically flat manifold with $N$ ends and interior boundary, is a manifold $M$, satisfying the properties described above, equipped with a Riemannian metric $g$ satisfying $(g-\og)\in W^{2,k}_{5/2-n}$, for some $k>n/2$.
\end{definition}
Note that the condition $k>n/2$ ensures that the metric is H\"older continuous, by the Sobolev-Morrey embedding.

By comparison to the Laplacian on a bounded domain, it is expected that boundary conditions must be enforced if we hope for $\Delta_g$, the Laplace-Beltrami operator associated with $g$, to be an isomorphism. We impose Dirichlet boundary conditions here, however Neumann boundary conditions could easily be used instead (cf. \cite{maxwell2005solutions}).

First note the following elementary estimate, which follows immediately from Proposition 1.6 of \cite{AF}.

\begin{lemma}\label{lemeasy}
Let $\delta\in\R$, then
\begin{equation}\label{nonfredholm}
\|u\|_{2,2,\delta}\leq C\left(\|\Delta_g u\|_{2,\delta-2}+\|u\|_{2,\delta}\right),
\end{equation}
for any $u\in H^2_\delta$.
\end{lemma}
Note that $\delta<\epsilon$ is required for the embedding $W^{k,p}_\delta\hookrightarrow W^{j,p}_{\epsilon}$ to be compact (cf. Lemma 2.1 of \cite{ellipticsys}), in addition to the usual condition $k>j$; that is, the estimate above does not suffice to prove Fredholmness. For this, we require Lemma \ref{lemineq}, below.

\begin{lemma}\label{lemineq}
Let $(M,g)$ be an asymptotically flat manifold as described above, and fix $\delta\in(2-n,0)$. Then for $u\in \Hdir$ we have
\begin{equation}\label{coercivelaplace}
\|u\|_{2,2,\delta}\leq C\|\Delta_g u\|_{2,\delta-2}.
\end{equation}
\end{lemma}
\begin{proof}
First note that $\Delta_g$ is asymptotic to the background Laplacian in the sense of \cite{AF} (Definition 1.5). Further note that the proof of Theorem 1.10 of \cite{AF} remains valid on an asymptotically flat manifold with boundary, so we have the scale-broken estimate,
\begin{equation}
\|u\|_{2,2,\delta}\leq C(\|\Delta_g u\|_{2,\delta-2}+\|u\|_{2,0}),
\end{equation}
which does indeed suffice to prove Fredholmness.

From which, we prove (\ref{coercivelaplace}) using a standard argument. Assume, to the contrary, that there exists a sequence $u_i$ such that $\|u_i\|_{2,2,\delta}=1$ and $\Delta_g u_i\rightarrow 0$. Passing to a subsequence, $u_i$ converges weakly in $H^2_\delta$ and by the weighted Rellich compactness theorem it converges strongly in $L^2_0$. Now (\ref{nonfredholm}) implies $u_i$ is Cauchy and therefore converges in $H^2_\delta$. By continuity, we have $\Delta_g u=0$, and therefore we have a non-trivial element of $\ker(\Delta_g)$. However, it can be seen directly from the maximum principle that $\Delta_g$ has trivial kernel in ${\Hdir}$. 
\end{proof}

From Lemma \ref{lemineq}, we establish the following.
\begin{prop}\label{propisolaplace}
For any $\delta\in(2-n,0)$, the map $\Delta_{g}:H^2_\delta\cap\overline{H}^1_\delta(M) \rightarrow L^2_{\delta-2}(M)$ is an isomorphism.
\end{prop}

\begin{proof}
We simply must prove that $\Delta_g$ is surjective, which is achieved by proving the range is closed and $\Delta_g^*$ has trivial kernel. It is a fairly standard argument to demonstrate that $\Delta_g$ has closed range, which is as follows. Take a sequence $u_i\in H^2_\delta\cap \bH^1_\delta(M)$ such that $\phi_i=\Delta_g u_i$ is Cauchy; that is, any Cauchy sequence in the range. By (\ref{coercivelaplace}), $u_i$ is convergent to some $u$, and by continuity, $\phi_i\rightarrow\Delta_gu$. It follows that $\Delta_g$ has closed range.

It remains to prove that $\Delta_g^*$ has trivial kernel. Note that the adjoint here does not equal the formal $L^2$ adjoint, but rather we interpret the equation $\Delta_g^*v=0$ in the weak sense:
$$\int_M \Delta_g (f) v \,dV=0$$
for all $f\in H^2_\delta\cap \bH^1_\delta(M)$, and from this standard elliptic regularity theory implies $v\in H^2_{\text{loc}}$. In particular, for any $\Omega\subset\joinrel\subset M$, we have
$$\int_\Omega \Delta_g (f)v\, dV=\int_\Omega f\Delta_g(v)\, dV=0 $$
for all $f\in C^\infty_c(\Omega)$, and therefore $\Delta v=0$ on $M$. It then follows that
$$\int_M \Delta_g (f)v\, dV=0=\int_{\partial M}\nabla (f)v\cdot dS$$
for all $f\in H^2_\delta\cap \bH^1_\delta(M)$, and therefore $v\equiv0$ on $\partial M$. Since $v\in L^2_{-5/2}$ is $H^2_{\text{loc}}$ and vanishes on $\partial M$, $v$ vanishes everywhere, again by the maximum principle.
\end{proof}

%%%%%%%%%%%%%%%%%%%%%%%%%%%%%%%%%%%%%%%%%%%%%%%%%%%%%%%%
\section{The phase space}\label{Sphasespace}

In this section we adapt Bartnik's phase space construction to an asymptotically flat manifold with an interior boundary. In particular, we show that the set of asymptotically flat initial data, with $g$ fixed on the boundary, is a Hilbert submanifold of the phase space. For simplicity, we restrict ourselves to the physically relevant case, $n=3$. Several of the results in the case considered here follow by entirely identical arguments as used by Bartnik, so we simply refer to the appropriate places in Ref. \cite{phasespace} for proofs in these instances. In addition to this, many proofs given here involve only small modifications to those given by Bartnik.

The constraint map is given by
\begin{align}
\Phi_0(g,\pi)&=R(g)\sqrt{g}-(\pi^{ij}\pi_{ij}-\frac{1}{2}(\pi^k_k)^2)/\sqrt{g},\\
\Phi_i(g,\pi)&=2\nabla_k\pi^k_i,
\end{align}
where $\sqrt{g}=\frac{\sqrt{\det g}}{\sqrt{\det \og}}$ is a volume form, and $\pi$ is related to the second fundamental form $K$, by $\pi^{ij}=(K^{ij}-g^{ij}\tr_g K)\sqrt{g}$. For a given energy-momentum source $(s,S^i)$, the constraint equations are $\Phi(g,\pi)=(s,S^i)$; in particular, the vacuum constraints are simply $\Phi(g,\pi)=0$.

Now let $(M,\og)$ be an asymptotically flat manifold as described in Section \ref{Slaplace}, where $\og$ will serve as a background metric. As we are motivated by considering extensions to a given compact manifold with boundary, $\Omega$, one should consider $\og$ near $\Sigma$ as coming from the metric on $\Omega$, which is to be extended. More concretely, one may choose $M$ such that it can be glued to $\Omega$ along $\Sigma$, and $\og$ would then be a smooth extension of the metric on $\Omega$. However, we avoid further discussion on $\Omega$ by simply considering $\og$ to be some given background metric. We define the domain and codomain of $\Phi$ in terms of weighted Sobolev spaces:
$$\mathcal{G}:=\{g\in S_2:g>0, (g-\og)\in \Hdirh(M)\},$$
$$\mathcal{K}:=H^1_{-3/2}(S^2\otimes \Lambda^3),\qquad\mathcal{N}:=L^2_{-5/2}(\Lambda^3\times T^*M\otimes\Lambda^3),$$
where $\Lambda^k$ is the space of $k$-forms on $M$, and $S_2$ and $S^2$ are symmetric covariant and contravariant 2-tensors on $M$ respectively. The phase space is the set of prospective initial data, $\mathcal{F}=\mathcal{G}\times\mathcal{K}$. The proofs of Proposition 3.1 and Corollary 3.2 of \cite{phasespace} apply directly in the case considered here, and it therefore follows immediately that $\Phi:\mathcal{F}\rightarrow\mathcal{N}$ is a smooth map of Hilbert manifolds.

It is interesting to note that at the time of publication, Bartnik's phase space concerned initial data that was slightly too rough to apply known results on the well-posedness of the Cauchy problem; however, through the positive resolution of the bounded $L^2$ curvature conjecture, Klainerman, Rodnianski and Szeftel \cite{boundedl2} have recently improved the local existence and uniqueness results to the case considered by Bartnik, and indeed the case considered here.

The key to proving that the level sets of $\Phi$ are Hilbert submanifolds, is a standard implicit function theorem style argument. As such, we study the linearisation of $\Phi$, which at at a point $(g,\pi)\in\mathcal{F}$, is given by
\begin{align}
D\Phi_{0\,(g,\pi)}[h,p]=&\,(\pi^k_k\pi^{ij}-2\pi^{ik}\pi^j_k)h_{ij}+\tr (h)(\frac{1}{2}\pi\cdot\pi-\frac14(\tr\pi)^2) /\sqrt{g}\nonumber\\
&+(\frac12\tr(h)R-\Delta\tr(h)+\nabla^i\nabla^jh_{ij}-R^{ij}h_{ij})\sqrt{g}\nonumber\\
&+(\tr(p)\tr(\pi)-2\pi\cdot p)/\sqrt{g}\\
D\Phi_{i\,(g,\pi)}[h,p]=&\,2\nabla_j(\pi^{jk}h_{ik})-\pi^{jk}\nabla_i h_{jk}+2\nabla_j p_i^j\label{Dconstraint2},
\end{align}
for $(h,p)\in T_{(g,\pi)}\mathcal{F}$. The formal $L^2$ adjoint is then computed as
\begin{align}
D\Phi^F_{1\,(g,\pi)}[N,X]=&\,N\left(\pi^k_k\pi^{ij}-2\pi^{ik}\pi^k_k+(\frac12\pi^{kl}\pi_{kl}-\frac14(\pi^k_k)^2)g^{ij}\right)/\sqrt{g}+\mathcal{L}_X\pi^{ij}\nonumber\\
&+\left(N(\frac12Rg^{ij}-R^{ij})+\nabla^i\nabla^kN-g^{ij}\nabla^k\nabla_kN\right)\sqrt{g}\label{adjointg}\\
D\Phi^F_{2\,(g,\pi)}[N,X]=&\,N(g_{ij}\pi^k_k-2\pi_{ij})/\sqrt{g}-\mathcal{L}_Xg_{ij}\label{adjointpi},
\end{align}
where $(N,X)\in \mathcal{N}^*=L^2_{-5/2}(\Lambda^0\times TM)$ and $\mathcal{L}$ is the Lie derivative on $M$. Note that we use the superscript `$F$' for the formal adjoint, rather than `$*$', which we reserve for the true adjoint.

We first give a coercivity estimate for $D\Phi_{(g,\pi)}^F$. It should be noted that this is simply Bartnik's Proposition 3.3 of \cite{phasespace}; however, particularly since there is a minor modification to the proof at the end, there is no harm in presenting the computation here. Furthermore, there is a minor omission in the argument of Bartnik that relies on a local version of this estimate, which we address in the proof of Proposition \ref{propweakstrong}. Note that for simplicity of presentation, we write $\xi=(N,X)$, which may be interpreted as a 4-vector in the spacetime.
\begin{prop}\label{propdphibound}
For all $\xi\in W^{2,2}_{-1/2}$, $D\Phi^F_{(g,\pi)}$ satisfies,
\begin{equation}
\|\xi\|_{2,2,-1/2}\leq C\left(\| D\Phi^F_{1\,(g,\pi)}[\xi]\|_{2,-5/2}+\|D\Phi^F_{2\,(g,\pi)}[\xi]\|_{1,2,-3/2}+\|\xi\|_{2,0}\right).
\end{equation}
\end{prop}
\begin{proof}
We will need to make use of the difference of connections tensor,
$$ \tilde{\Gamma}=\Gamma-\mathring{\Gamma}=\frac{1}{2}g^{il}(\onabla_j g_{lk}+\onabla_k g_{jl}-\onabla_l g_{jk}),$$
which is clearly controlled in $W^{1,2}_{-3/2}$, for $g\in\mathcal{G}$.

Rearranging (\ref{adjointg}) gives
\begin{equation*}
\nabla^i\nabla^j N-g^{ij}\nabla^k\nabla_k N=S^{ij},
\end{equation*}
where $S$ is given by
\begin{align*}
\sqrt{g}S^{ij}=&\,D\Phi_{g}^*[\xi]^{ij}-N\left(\pi^k_k\pi^{ij}-2\pi^{ik}\pi^j_k-\Big( N(\frac{1}{2}Rg^{ij}-R^{ij})\Big)\sqrt{g}\right.\\
&\left.+\Big(\frac{1}{2}\pi^{kl}\pi_{kl}-\frac{1}{4}(\pi^k_k)^2\Big) g^{ij}\right)/\sqrt{g} +\mathcal{L}_X\pi^{ij}.
\end{align*}
From this, we can then write
\begin{equation}
\nabla^i\nabla^j N=S^{ij}-\frac{1}{2}g^{ij}S^k_k,
\end{equation}
which gives an estimate for $\nabla^2N$:
\begin{equation}
\|\nabla^2 N\|_{2,-5/2}\leq C\|S\|_{2,-5/2}.
\end{equation}
Noting that $(g,\pi)$ is fixed and $\xi=(N,X)$, the standard weighted Sobolev-type inequalities give
\begin{align*}
\|\onabla^2N\|_{2,-5/2}\leq &\,C \Big{(}\|D\Phi^F_{1\,(g,\pi)}[\xi]\|_{2,-5.2}+\|\tilde{\Gamma}\onabla N\|_{2,-5/2}+\|\pi\onabla X\|_{2,-5/2}\\
&+\|X\onabla\pi\|_{2,-5/2}+\|N\|_{\infty,0}(\|\pi^2\|_{2,-5/2}+\|Ric(g)\|_{2,-5/2})\Big{)}\\
\leq & \,C\left(\|D\Phi^F_{1\,(g,\pi)}[\xi]\|_{2,-5/2}+\|\xi\|_{\infty,0}+\|\onabla\xi\|_{3,-1}(\|\tilde{\Gamma}\|_{6,-3/2}+\|\pi\|_{6,-3/2})\right)\\
\leq & \,C\left(\|D\Phi^F_{1\,(g,\pi)}[\xi]\|_{2,-5/2}+\|\xi\|_{\infty,0}+\|\onabla\xi\|_{3,-1}(\|\tilde{\Gamma}\|_{1,2,-3/2}+\|\pi\|_{1,2,-3/2})\right)\\
\leq & \,C\left(\|D\Phi^F_{1\,(g,\pi)}[\xi]\|_{2,-5/2}+\|\xi\|_{\infty,0}+\|\onabla\xi\|_{3,-1}\right).
\end{align*}

The Bianchi identity, the identity $R_{ijkl}X^l=\nabla_i\nabla_j X_k-\nabla_j\nabla_i X_k$, and a bit of algebraic manipulation result in
$$\nabla_k\mathcal{L}_X g_{ij}+\nabla_j\mathcal{L}_X g_{ik}-\nabla_i\mathcal{L}_X g_{jk}=2(R_{ikjl}X^l+\nabla_k\nabla_j X_i),$$
and therefore can estimate $\nabla^2X$ by
\begin{equation}
\|\nabla^2X\|_{2,-5/2}\leq C(\|Riem(g)\|_{2,-5/2}\|X\|_{\infty,0}+\|\nabla \mathcal{L}_X g\|_{2,-5/2}).
\end{equation}
By writing the Riemann tensor explicitly in terms of $g,\onabla g$ and $\onabla^2 g$, it is clear that we can control $\|Riem(g)\|_{2,-5/2}$ for $g\in\mathcal{G}$; the Riemann tensor is quadratic in $\onabla g$ and linear in $\onabla^2 g$\footnote{This entirely straightforward, albeit aesthetically unpleasing, computation can be found, for example, in Appendix A of \cite{mythesis}.}.

Making use of (\ref{adjointpi}), the Lie derivative is expressed as
\begin{equation}
\mathcal{L}_X g_{ij}=N(g_{ij}\pi^k_k-2\pi_{ij})g^{-1/2}-D\Phi^F_{2\,(g,\pi)}[\xi]_{ij},
\end{equation}
and from this, the weighted Sobolev-type inequalities give
\begin{align*}
\|\nabla \mathcal{L}_X g\|_{2,-5/2}\leq &C  \,(\|\nabla(N\pi)\|_{2,-5/2}+\|\nabla D\Phi^F_{2\,(g,\pi)}[\xi]\|_{2,-5/2})\\
\leq &C \, \Big{(}\|\onabla D\Phi^F_{2\,(g,\pi)}[\xi]\|_{2,-5/2}+\|\tilde{\Gamma}\|_{4,-1}\|D\Phi^F_{2\,(g,\pi)}[\xi]\|_{4,-3/2}\\
&+\|\onabla N\|_{3,-1}\|\pi\|_{6,-3/2}+\|N\|_{\infty,0}(\|\onabla\pi\|_{2,-5/2}+\|\tilde{\Gamma}\pi\|_{2,-5/2})\Big{)}\\
\leq &C  \, \Big{(}\|\onabla D\Phi^F_{2\,(g,\pi)}[\xi]\|_{2,-5/2}+\|\tilde{\Gamma}\|_{1,2,-3/2}\|D\Phi^F_{2\,(g,\pi)}[\xi]\|_{1,2,-3/2}\\
&+\|\onabla N\|_{3,-1}\|\pi\|_{1,2,-3/2}+\|N\|_{\infty,0}(\|\onabla\pi\|_{2,-5/2}+\|\tilde{\Gamma}\|_{1,2.-3/2}\|\pi\|_{1,2,-3/2})\Big{)}\\
\leq &C  \, \big{(}\|D\Phi^F_{2\,(g,\pi)}[\xi]\|_{1,2,-3/2}+\|N\|_{\infty,0}+\|\onabla N\|_{3,-1}\big{)}.
\end{align*}
We now obtain an estimate for $\|\onabla^2 X\|$ in terms of $\|\nabla^2 X\|$ as follows:
\begin{align*}
\|\onabla^2X\|_{2,-5/2}\leq C&  \, \big{(}\|\nabla^2 X\|_{2,-5/2}+\|\onabla (X)\tilde{\Gamma}\|_{2,-5/2}+\|X\onabla(\tilde{\Gamma})\|_{2,-5/2}\\
&+\|\tilde{\Gamma}^2X\|_{2,-5/2}\big{)}\\
\leq &C \, \big{(}\|\nabla^2 X\|_{2,-5/2}+\|\onabla X\|_{3,-1}\|\tilde{\Gamma}\|_{6,-3/2}\\
&+\|X\|_{\infty,0}(\|\onabla\tilde{\Gamma}\|_{2,-5/2}+\|\tilde{\Gamma}^2\|_{2,-5/2})\big{)}\\
\leq &C \, \big{(}\|\nabla^2 X\|_{2,-5/2}+\|\onabla X\|_{3,-1}\|\tilde{\Gamma}\|_{1,2,-3/2}\\
&+\|X\|_{\infty,0}(\|\onabla\tilde{\Gamma}\|_{2,-5/2}+\|\tilde{\Gamma}\|^2_{1,2,-3/2})\big{)}\\
\leq &C  \, \big{(}\|\nabla^2 X\|_{2,-5/2}+\|\onabla X\|_{3,-1}+\|X\|_{\infty,0}\big{)}.
\end{align*}

Combining these we have
\begin{equation}
\|\onabla^2\xi\|_{2,-5/2}\leq C\left(\|D\Phi^F_{1\,(g,\pi)}[\xi]\|_{2,-5/2}+\|D\Phi^F_{2\,(g,\pi)}[\xi]\|_{1,2,-3/2}+\|\xi\|_{\infty,0}+\|\onabla\xi\|_{3,-1}\right).\label{coercive1}
\end{equation}
The last two terms on the right-hand side are estimated using the weighted inequalities, Young's inequality, and the definition of the $W^{k,p}_\delta$ norm directly:
\begin{align}
\|\xi\|_{\infty,0}&\leq c \|\xi\|_{1,4,0}=\|\xi^{1/4}\xi^{3/4}\|_{1,4,0}\nonumber\\
& \leq c \|\xi^{1/4}\|_{1,8,0}\|\xi^{3/4}\|_{1,8,0}\nonumber\\
& \leq c \|\xi\|^{1/4}_{1,2,0}\|\xi\|^{3/4}_{1,6,0}\nonumber\\
& \leq c  \|\xi\|^{1/4}_{1,2,0}\|\xi\|^{3/4}_{2,2,0}\label{infty0bound}\\
& \leq c \|\xi\|^{1/4}_{1,2,0}(\|\xi\|_{1,2,0}+\|\onabla^2\xi\|_{2,-2})^{3/4}\nonumber\\
& \leq c \epsilon^{-3}\|\xi\|_{1,2,0}+\epsilon(\|\xi\|_{1,2,0}+\|\onabla^2\xi\|_{2,-2})\nonumber\\
& \leq c \epsilon^{-3}\|\xi\|_{1,2,0}+\epsilon\|\onabla^2\xi\|_{2,-2},\nonumber
\end{align}
for any $\epsilon>0$.

An estimate for the final term in (\ref{coercive1}) is obtained almost identically:
\begin{align}
\|\onabla\xi\|_{3,-1}&\leq  \|\xi\|_{1,3,0}=\|\xi^{1/3}\xi^{2/3}\|_{1,3,0}\nonumber\\
& \leq c \|\xi^{1/3}\|_{1,6,0}\|\xi^{2/3}\|_{1,6,0}\nonumber\\
& \leq c \|\xi\|^{1/3}_{1,2,0}\|\xi\|^{2/3}_{1,4,0}\nonumber\\
& \leq c  \|\xi\|^{1/3}_{1,2,0}\|\xi\|^{2/3}_{2,2,0}\label{31bound}\\
& \leq c \|\xi\|^{1/3}_{1,2,0}(\|\xi\|_{1,2,0}+\|\onabla^2\xi\|_{2,-2})^{2/3}\nonumber\\
& \leq c \epsilon^{-2}\|\xi\|_{1,2,0}+\epsilon(\|\xi\|_{1,2,0}+\|\onabla^2\xi\|_{2,-2})\nonumber\\
& \leq c \epsilon^{-2}\|\xi\|_{1,2,0}+\epsilon\|\onabla^2\xi\|_{2,-2}\nonumber.
\end{align}
By inserting these estimates back into (\ref{coercive1}), we obtain
$$\|\onabla^2\xi\|_{2,-5/2}\leq C \big{(}\|D\Phi^F_{1\,(g,\pi)}[\xi]\|_{2,-5/2}+\|D\Phi^F_{2\,(g,\pi)}[\xi]\|_{1,2,-3/2}\big{)}+c(\epsilon)\|\xi\|_{1,2,0}+\epsilon\|\onabla^2\xi\|_{2,-2};$$
choosing $\epsilon$ small enough and applying the interpolation inequality gives
\begin{equation}
\|\onabla^2\xi\|_{2,-5/2}\leq C \left(\|D\Phi^F_{1\,(g,\pi)}[\xi]\|_{2,-5/2}+\|D\Phi^F_{2\,(g,\pi)}[\xi]\|_{1,2,-3/2}+\|\xi\|_{2,0}\right).
\end{equation}
Up to this point, we have essentially reproduced Bartnik's argument, albeit with slightly more detail, and if we had a weighted Poincar\'e ineqality we would be done; however, we are unaware of an appropriate Poincar\'e inequality in the case of a general asymptotically flat manifold with an interior boundary. Instead we consider separately the inequality near infinity, where we do have an appropriate Poincar\'e inequality, and on a compact domain. For some exterior region $E_{R_0}$ we have the Poincar\'e inequality and therefore it follows that we have
$$\|\xi\|_{2,2,-1/2}\leq C \left(\|D\Phi^F_{1\,(g,\pi)}[\xi]\|_{2,-5/2}+\|D\Phi^F_{2\,(g,\pi)}[\xi]\|_{1,2,-3/2}+\|\xi\|_{2,0}+\|\xi\|_{1,2:M\setminus E_{R_0}}\right).$$
Applying the interpolation inequality again and noting $\|\xi\|_{2:M\setminus E_{R_0}}\leq C\|\xi\|_{2,0}$ completes the proof.
\end{proof}

\begin{remark}\label{remark1}
While Proposition \ref{propdphibound} gives an estimate on $M$, the weighted H\"older, Sobolev and interpolation inequalities used above are also valid on an annular region $A_R:=\{x\in M:r(x)\in(R,2R)\}$ (cf. \cite{AF}). In particular, we have
\begin{equation}\label{annular}
\|\onabla^2\xi\|_{2,-5/2:A_R}\leq C\Big{(}\|D\Phi^F_{1\,(g,\pi)}[\xi]\|_{2,-5/2:A_R}+\|D\Phi^F_{2\,(g,\pi)}[\xi]\|_{1,2,-3/2:A_R}+\|\xi\|_{2,0:A_R}\Big{)}
\end{equation}
for $\xi\in W^{2,2}_{\delta}(A_R)$, where $C$ is independent of $R$. However, we do not have the same control on $\|\xi\|_{2,2,-1/2:A_R}$, as the constant in the Poincar\'e inequality depends on $A_R$.
\end{remark}

Note that the true adjoint of the linearised constraint map, $D\Phi^*_{(g,\pi)}$, is only defined in the weak sense, which is why we make the distinction between $D\Phi^*_{(g,\pi)}$ and $D\Phi^F_{(g,\pi)}$. In order to study the kernel of $D\Phi^*_{(g,\pi)}$ we must first demonstrate that weak solutions to the equation $D\Phi^*_{(g,\pi)}[\xi]=0$ are sufficiently regular to consider this as a bona fide differential equation.

\begin{prop}\label{propweakstrong}
Suppose $\xi\in \mathcal{N}$ is a weak solution of $D\Phi_{(g,\pi)}^*[\xi]=(f_1,f_2)$, where ${(f_1,f_2)\in L^2_{-5/2}\times W^{1,2}_{-3/2}}$ and $(g,\pi)\in\mathcal{F}$, then $\xi\in\Hdirh$ and furthermore, ${D\Phi^*_{(g,\pi)}[\xi]=D\Phi^F_{(g,\pi)}[\xi]}$.
\end{prop}
\begin{proof}
We first note that local regularity follows directly from Bartnik's proof of Proposition 3.5 in Ref. \cite{phasespace}. The only possible place in Bartnik's proof where the boundary terms may come in to play are in choosing $(h,p)$ supported in some coordinate neighbourhood. Clearly our boundary conditions do not prevent this, so there is no obstruction to applying Bartnik's proof directly. That is, $\xi\in H^2_{loc}$.

In the following, let $B_R$ be an open ``ball" of radius $R$; for $R>2$, $B_R:=\{x\in M: r(x)<R\}$, and define $M_{\epsilon R}:=\{x\in B_R: \text{dist}(\Sigma,x)>\epsilon\}$, for some small $\epsilon$.

$$\int_{M_{\epsilon R}}D\Phi_{(g,\pi)}[h,p]\cdot\xi=\int_{M_{\epsilon R}}(h,p)\cdot(f_1,f_2)$$
for all $(h,p)\in C^\infty_c(M_{\epsilon R})$. In particular, since $\xi\in H^2_{-1/2}(M_{\epsilon R})$, we have
$$\int_{M_{\epsilon R}}(h,p)\cdot D\Phi^F_{(g,\pi)}[\xi]=\int_{M_{\epsilon R}}(h,p)\cdot(f_1,f_2);$$
that is, $D\Phi^F_{(g,\pi)}[\xi]=(f_1,f_2)$ on any $M_{\epsilon R}$. We then have $D\Phi^*_{(g,\pi)}[\xi]=D\Phi^F_{(g,\pi)}[\xi]$ on $M$; that is, the formal adjoint is indeed the true adjoint when $(f_1,f_2)\in L^2_{-5/2}\times H^1_{-3/2}$, as expected.

It remains to demonstrate that $\xi$ satisfies the boundary conditions and exhibits the correct asymptotics. To this end, we introduce a new smooth cutoff function $\chi\in C^\infty_c(M)$ such that $\chi\equiv1$ on $B_{R_0}$, for some $R_0>2$ and $\chi=0$ on $B_{2R_0}$. Define $\chi_R(x)=\chi(xR_0/R)$, so that $\chi_R$ has support on $B_{2R}$. Clearly $\chi_R\xi\in W^{2,2}_{-1/2}$, therefore Proposition \ref{propdphibound} gives
\begin{align}
\|\chi_R\xi\|_{2,2,-1/2}\leq &\,C\Big{(}\|D\Phi_1^*[\chi_R\xi]\|_{2,-5/2}+\|D\Phi_2^*[\chi_R\xi]\|_{1,2,-3/2}+\|\xi\|_{2,0}\Big{)}\label{chiRbound},
\end{align}
noting that $\chi_R\xi\rightarrow\xi$ in $L^2_{0}$. From this we can show that $\chi_R\xi$ is uniformly bounded in $W^{2,2}_{-1/2}$. Obtaining control of $\|\chi_R\xi\|_{2,2,-1/2}$ independent of $R$ is the minor omission in Ref. \cite{phasespace} mentioned above, however this is easily resolved as follows.

Note that $\onabla\chi_R(x)=(R_0/R)\onabla\chi(xR_0/R)$, $\onabla\chi$ is bounded, and $\onabla\chi_R$ has support on $A_R$. It follows that we have
\begin{equation*}
\|u\onabla\chi_R\|_{p,\delta}\leq c\|u/R\|_{p,\delta:A_R}\leq c \sup_{x\in A_R}|r(x)/R|\|u\|_{p,\delta+1:A_R}\leq c\|u\|_{p,\delta+1:A_R}.
\end{equation*}
From this, the expression for $D\Phi^F$, and the usual weighted Sobolev-type inequalities, we have
\begin{align*}
\|D\Phi^F_1[\chi_R\xi]\|_{2,-5/2}\leq \, &c\, \Big{(}\|\chi_R D\Phi^F_1[\xi]\|_{2,-5/2}+\|\pi\xi\onabla\chi_R\|_{2,-5/2}\\
&+\|\xi\onabla^2\chi_R\|_{2,-5/2}+\|\onabla(\xi)\onabla(\chi_R)\|_{2,-5/2}\Big{)}\\
\leq\, &c \, \Big{(}\|f_1\|_{2,-5/2}+\|\pi\|_{4,-3/2}\|\xi\|_{4,0:A_R}+\|\xi\|_{2,-1/2}\\
&+\|\onabla\xi\|_{2,-3/2:A_R} \Big{)}\\
\leq\, &c \, \big{(}\|f_1\|_{2,-5/2}+\|\pi\|_{1,2,-3/2}\|\xi\|_{1,2,0:A_R}+\|\xi\|_{2,-1/2}\\
&+\|\onabla\xi\|_{2,-3/2:A_R} \Big{)}\\
\leq\, &C  \, (\|f_1\|_{2,-5/2}+\|\xi\|_{2,-1/2}+\|\onabla\xi\|_{2,-3/2:A_R})\\
\leq\, &C \, (\|f_1\|_{2,-5/2}+\|\xi\|_{2,-1/2})+\epsilon\|\onabla^2\xi\|_{2,-5/2:A_R}.
\end{align*}
Almost identically, we have
\begin{align*}
\|\onabla D\Phi^F_2[\chi_R\xi]\|_{2,-5/2}\leq\, & \, \Big{(}\|\chi_R\onabla D\Phi^F_2[\xi]\|_{2,-5/2}+\|\pi\xi\onabla\chi_R\|_{2,-5/2}\\
&+\|\xi\onabla^2\chi_R\|_{2,-5/2}+\|\onabla(\xi)\onabla(\chi_R)\|_{2,-5/2}\Big{)}\\
\leq\,& C \, (\|\onabla f_2\|_{2,-5/2}+\|\xi\|_{2,-1/2})+\epsilon\|\onabla^2\xi\|_{2,-5/2:A_R}
\end{align*}
and a similar estimate for $\|D\Phi^F_2[\chi_R\xi]\|_{2,-3/2}$ holds, so we have in fact,
$$\|D\Phi^F_2[\chi_R\xi]\|_{1,2,-3/2}\leq C  (\|\onabla f_2\|_{2,-5/2}+\|\xi\|_{2,-1/2})+\epsilon\|\onabla^2\xi\|_{2,-5/2:A_R}.$$

Inserting the estimates above into (\ref{chiRbound}) we arrive at
\begin{align}
\|\onabla^2(\chi_R\xi)\|_{2,-5/2}\leq &\,C  \big{(}\|f_1\|_{2,-5/2}+\|f_2|_{1,2,-3/2}\nonumber\\
&+\|\xi\|_{2,-1/2}+\epsilon\|\onabla^2\xi\|_{2,-5/2:A_R}\big{)}.\label{dummypi}
\end{align}
Unfortunately we are unable to ensure $\|\onabla^2\xi\|_{2,-5/2:A_R}\lesssim\|\onabla^2(\chi_R\xi)\|_{2,-5/2}$, so we can not absorb the last term into the left-hand side of (\ref{dummypi}). Recalling Remark \ref{remark1}, we apply the local version of Proposition \ref{propdphibound} to obtain
\begin{equation}
\|\onabla^2\xi\|_{2,-5/2:A_R}\leq C\|f_1\|_{2,-5/2}+\|f_2\|_{1,2,-3/2}+\|\xi\|_{2,0}.
\end{equation}
Finally we obtain the desired uniform bound:
\begin{align}
\|\chi_R\xi\|_{2,2,-1/2}&\leq C(\|\chi_R\xi\|_{2,-1/2}+\|\onabla^2(\chi_R\xi)\|_{2,-5/2})\nonumber\\
&\leq \,C  \big{(}\|f_1\|_{2,-5/2}+\|f_2\|_{1,2,-3/2}+\|\xi\|_{2,-1/2}\big{)}.
\end{align}
It follows that $\chi_R\xi$ converges weakly to $\xi$ in $H^2_{-1/2}$. Now, since the formal adjoint agrees with the true adjoint, the boundary terms arising from integration by parts necessarily vanish; explicitly (cf. eq. (\ref{boundaryterms})),
\begin{equation}
\oint_\Sigma\left(\xi^0(\nabla_i\text{tr}_gh-\nabla^jh_{ij})\sqrt{g}-2\xi^jp_{ij}\right)dS^i=0,
\end{equation}
for all $(h,p)\in(\Hdirh)\times H^1_{-3/2}$. It follows that $\xi$ vanishes on $\Sigma$ and therefore, $\xi\in\Hdirh$.
\end{proof}

\begin{theorem}\label{thmlevelset}
For all $(s,S)\in\mathcal{N}$, the level set $\mathcal{C}(s,S):=\Phi^{-1}(s,S)$ is a Hilbert submanifold of $\mathcal{F}$. We refer to this as the constraint manifold.
\end{theorem}
\begin{proof}
By the implicit function theorem, we simply must demonstrate that $D\Phi_{(g,\pi)}$ is surjective and the kernel splits. The kernel trivially splits with respect to the Hilbert structure, so we simply must prove that $D\Phi^*_{(g,\pi)}$ has trivial kernel and $D\Phi_{(g,\pi)}$ has closed range. It is clear from the above, that elements in the kernel of $D\Phi^*_{(g,\pi)}$ indeed satisfy $D\Phi^F_{(g,\pi)}=0$. Once we have this, note that Bartnik's proof of the triviality of $\ker(D\Phi^F_{(g,\pi)})$ relies only on the structure of the equation and the asymptotics assumed\footnote{The proof essentially makes use of the asymptotics to show that any element of the kernel must be supported away from infinity, then shows if an element of the kernel vanishes on a portion of a small ball then it vanishes on the entire ball. By covering $M$ with balls of this (fixed) size, and noting $M$ is connected, the conclusion follows. It is clear a boundary has no impact on this argument.} -- it is entirely unaffected by the inclusion of an interior boundary. Therefore this proof applies here and we simply must prove that $D\Phi_{(g,\pi)}$ is surjective, which is again adapted from Bartnik's arguments to deal with the boundary. The key to making this argument work is the estimate given earlier by Lemma \ref{lemineq}.

The idea is to consider a restriction of $D\Phi_{(g,\pi)}$ to variations of a particular form, so that the operator becomes elliptic. Then we simply must show that this restricted operator has closed range and finite dimensional cokernel. We consider
$$ h_{ij}(y)=-\frac12yg_{ij},\qquad p^{ij}(Y)=\frac12(\nabla^iY^j+\nabla^jY^i-g^{ij}\nabla_kY^k)\sqrt{g}$$
for $y,Y\in \Hdirh(M)$.

For the operator $F[y,Y]:=D\Phi_{(g,\pi)}[h(y),p(Y)]$, we have
$$F[y,Y]=
\begin{bmatrix}
\Delta y\sqrt{g}-\frac{1}{4}\Phi_0(g,\pi)y+\frac{1}{2}\pi^k_k\nabla_j Y^j-2\pi^{ij}\nabla_i Y_j\\
\Delta Y_i \sqrt{g}+R_{ij}Y^j\sqrt{g}-\nabla_j(\pi^j_i)y-\pi_i^j\onabla_jy+\frac{1}{2}\pi^j_j\onabla_iy
\end{bmatrix}.
$$
and the formal adjoint is given by
$$F^F[z,Z]=
\begin{bmatrix}
\Delta z\sqrt{g}-\frac{1}{4}\Phi_0(g,\pi)z+\pi^j_i\onabla_j Z^i-\frac{1}{2}\onabla_i(\pi^j_j Z^i)\\
\Delta Z_j \sqrt{g}+2\nabla_i(\pi^i_j z)-\frac{1}{2}\nabla_j(\pi^i_i z)+R_{ij}Z^i\sqrt{g}
\end{bmatrix}.
$$

It follows from the proof of Proposition \ref{propweakstrong}, that $(z,Z)$ satisfying $F^*[z,Z]=0$ are $H^2_{-1/2}$ and the boundary terms arrising from integration by parts vanish; that is, $(z,Z)\in\Hdirh(M)$. From Lemma \ref{lemineq}, it is straightforward to show using the weighted H\"older, Sobolev and interpolation inequalities (cf. eq. (3.42) of \cite{phasespace}), that we have the scale-broken estimate:
\begin{equation}\label{scalebrokenF}
\|(y,Y)\|_{2,2,-1/2}\leq C(\|F[y,Y]\|_{2,-5/2}+\|(y,Y)\|_{2,0}).
\end{equation}
It is now a standard argument to demonstrate that $F$ has closed range and finite dimensional cokernel (cf. Ref. \cite{ellipticsys}, Theorem 6.3, and Ref. \cite{AF}, Theorem 1.10).

Let $(y,Y)_i$ be a sequence in $\ker(F)$ satisfying $\|(y,Y)\|_{2,2,-1/2}\leq 1$; that is, a sequence in the closed unit ball in $\ker(F)$. By the weighted Rellich compactness theorem, passing to a subsequence, $(y,Y)_{i_n}$ converges strongly in $L^2_0$, which in turn implies via (\ref{scalebrokenF}) that $(y,Y)_{i_n}$ converges strongly in $H^2_{-1/2}$. That is, the closed unit ball in $\ker(F)$ is compact, and therefore $\ker(F)$ is finite dimensional. It follows that the domain of $F$ can be split as $\Hdir=\ker(F)\oplus Z$, for some closed orthogonal complementary subspace, $Z$. Now, for $(y,Y)\in Z$, we prove
\begin{equation}\label{Zeq}
\|(y,Y)\|_{2,2,-1/2}\leq C\|F[y,Y]\|_{2,-5/2}
\end{equation}
by contradiction. Assume there is a sequence $(y,Y)_i\in W$ satisfying ${\|(y,Y)_i\|_{2,2,-1/2}=1}$, while $\|F[y,Y]_i\|_{2,-5/2}\rightarrow 0$. By the above argument, passing to a subsequence, we have that $(y,Y)_{i_n}$ converges strongly to $(y,Y)\in W$. By continuity, $F[y,Y]=0$, while $\|(y,Y)\|_{2,2,-1/2}=1$, implying the intersection of $\ker(F)$ and $W$ is nontrivial. That is, by contradiction, (\ref{Zeq}) holds. An identical argument to that used in the proof of Proposition \ref{propisolaplace} now shows that $F$ has closed range.

Furthermore, since $F^F$ has the same form as $F$, an estimate of the form of (\ref{scalebrokenF}) also holds for ${(z,Z)\in\ker(F^*)}$, which implies that $\ker(F^*)$ is finite dimensional. Since the range of $F$ is contained in the range of $D\Phi$, we have surjectivity of $D\Phi$ and therefore completes the proof.
\end{proof}

%%%%%%%%%%%%%%%%%%%%%%%%%%%%%%%%%%%%%%%%%
\section{Critical points of the ADM mass}\label{Scriticalpoints}
In \cite{phasespace} Bartnik discusses a result of Corvino, which states that if there exists an asymptotically flat extensions to a compact manifold with boundary, minimising the ADM energy, then it must be a static metric \cite{Corvino2000}. Specifically, Bartnik argues that it would be more natural to obtain Corvino's result from the Hamiltonian considerations he uses to prove a similar result for complete manifolds with no boundary. Here we give such an argument, considering the mass rather than the energy, and obtain that critical points of the mass functional only occur if the exterior is stationary. Furthermore, if these stationary solutions are sufficiently regular, they must in fact be static black hole exteriors. It should be noted that our set of extensions is larger than the usual set of admissible extensions in the context of the Bartnik mass. In order to ensure the validity of the positive mass theorem, we would also require conditions on the mean curvature of $\Sigma$ (see \cite{miao2004positive}); however it is not clear how to modify the arguments here to include mean curvature boundary conditions.

The content of this section has also been discussed in a recent note \cite{massmini} using stronger boundary conditions than considered here, and indeed stronger than the preferred mean curvature boundary conditions mentioned above; however, an analogous analysis to that in Section \ref{Sphasespace} was not given.

As in the preceding section, we quote Bartnik's results where the proofs require no modifications to this case. Furthermore, the results established here are again based on adapting Bartnik's arguments to deal with the boundary. The results of Section \ref{Sphasespace} are precisely what is needed for these arguments to work in the case considered here.

The energy-momentum covector $\mathbb{P}_\mu=(m_0,p_i)$ is defined by
\begin{align}
16\pi m_0&:=\oint_{S_\infty}\og^{jk}(\onabla_kg_{ij}-\onabla_i g_{jk})dS^i,\\
16\pi p_i&:=2\oint_{S_\infty}\pi_{ij}dS^j.
\end{align}

It is useful to consider the pairing of the energy-momentum vector with some asymoptotic translation, $\xi_\infty=(\xi_\infty^0,\xi_\infty^i)\in\R^{1+3}$,
$$16\pi\xi_\infty\cdot\mathbb{P}=\oint_\infty\left(\xi_\infty^0\og^{ik}(\onabla_k g_{ij}-\onabla_j g_{ik})+2\xi_\infty^i\pi_{ij}\right)dS^j.$$
By writing this as scalar-valued flux integral at infinity, we can make sense of this as an integral over $M$ through the divergence theorem. To extend $\xi_\infty$ to a scalar function and vector field over $M$, we identify $\xi_\infty^0$ with a constant function and $\xi^i_\infty$ with a $\og$-parallel vector field in a neighbourhood of infinity; that is, we identify $\xi_\infty$ with some $\tilde{\xi}$, defined near infinity and satisfying $\onabla\tilde\xi\equiv0$. We then choose any smooth bounded $\xiref=(\xiref^0,\xiref^i)$ supported away from $\Sigma$ and with $\xiref\equiv\tilde{\xi}$ near infinity to represent $\xi_\infty$. This allows us to write the energy-momentum as
\begin{align}
16\pi\xi^0_\infty\mathbb{P}_0(g)=&\int_M\Big{(}\xiref^0(\og^{ki}\og^{jl}\onabla_k\onabla_l g_{ij}-\mathring{\Delta}\text{tr}_{\mathring{g}} g)\nonumber\\
&+\og^{ki}\og^{jl}\onabla_k \xiref^0(\onabla_l g_{ij}-\onabla_i \og_{jl})\Big{)} \sqrt{\mathring{g}},\label{Edefna}\\
16\pi\xi^i_\infty\mathbb{P}_i(\pi)=&\,2\int_M\left(\xiref^i\onabla_j\pi_i^j+\pi^j_i\onabla_j\xiref^i\right).\label{pdefna}
\end{align}

Now it should be noted that $\mathbb{P}$ is not well-defined everywhere on $\mathcal{F}$; however, it is well-defined on any constraint manifold $\mathcal{C}(s,S)$ with $(s,S)\in L^1=L^1_{-3}$. In Section 4 of \cite{phasespace}, it is shown that this definition is equivalent to the usual definition of the ADM energy-momentum and is in fact a smooth map on each  $\mathcal{C}(s,S)$ with $(s,S)\in L^1$.

It is well known that the mass must be added to the ADM Hamiltonian in order to generate the correct equations of motion \cite{RT}. The formal equations of motion arising from the ADM Hamiltonian are indeed the correct evolution equations, however the boundary terms, coming from the integration by parts, correspond to the linearisation of the energy-momentum; that is, for $(g,\pi)\in\mathcal{F}$, the correct Hamiltonian to generate the equations of motion is given by
\begin{equation}\label{RTham}
\mathcal{H}^{(\xi)}(g,\pi)=16\pi \xi^\mu_\infty\mathbb{P}_\mu-\int_M\xi^\mu\Phi_\mu(g,\pi),
\end{equation}
where $\xi\in\Xi:=\{\xi:\xi-\xiref\in\Hdirh(M)\}$. While the separate terms in (\ref{RTham}) are not well-defined on $\mathcal{F}$,
by combining the terms into a single integrand, the dominant terms in each component cancel exactly (cf. \cite{phasespace}). Henceforth, we consider the Hamiltonian to be this regularised one, with the dominant terms canceled. Note that the boundary conditions imposed on $\xi$ are required to ensure that the surface integrals on $\Sigma$, due to integration by parts in obtaining the equations of motion, do indeed vanish. This can be seen by considering the following:
\begin{align}
(h&,p)\cdot D\Phi^F_{(g,\pi)}[\xi]-\xi\cdot D\Phi_{(g,\pi)}[h,p]\nonumber\\
=&\,\nabla^i\Big{(}(\xi^0(\onabla_i\text{tr}_gh-\nabla^jh_{ij})+\onabla^j(\xi^0)h_{ij}-\text{tr}_g h\onabla_i(\xi^0))\sqrt{g}-2\xi^jp_{ij}\Big{)}\nonumber\\
&-\nabla^i\Big{(}2\pi^k_i h_{jk}\xi^j - \pi^{jk}h_{jk}\xi_i \Big{)}.\label{boundaryterms}
\end{align}
The surface integrals at infinity are exactly cancelled by the term $16\pi \xi^\mu_\infty\mathbb{P}_\mu$ (cf. \cite{phasespace}). In particular, we have for all $(g,\pi)\in\mathcal{F}$, $(h,p)\in T_{(g,\pi)}\mathcal{F}$ and $\xi\in\Xi$,
\begin{equation}\label{DHphi}
D\mathcal{H}^{(\xi)}_{(g,\pi)}[h,p]=-\int_M(h,p)\cdot D\Phi^F_{(g,\pi)}[\xi].
\end{equation}
The ability to express the variation of the Hamiltonian in this form is precisely what we mean by the statement that the correct equations of motion are generated. In this form, we can interpret the variation of the Hamiltonian density with respect to each of $g$ and $\pi$; that is, $\frac{\delta H^{(\xi)}}{\delta g}=D\Phi^F_{1\,(g,\pi)}[\xi]$. We then can write Hamilton's equations as
\begin{equation}\label{evoeq}
\frac{\partial}{\partial t}
\begin{bmatrix}
g\\ \pi
\end{bmatrix}
=-\begin{bmatrix}0&1\\-1&0\end{bmatrix}\circ D\Phi_{(g,\pi)}^F[\xi],
\end{equation}
where $t$ is interpreted as the flow parameter of $(N,X)$ in the full spacetime; this is exactly the Einstein evolution equations. This also motivates a result of Moncrief \cite{moncrief1}, equating solutions to $D\Phi_{(g,\pi)}^F[\xi]=0$ with Killing vectors in the spacetime. For this reason, we say an initial data set $(g,\pi)$ is stationary if there exists $\xi$, asymptotic to a constant timelike translation, satisfying $D\Phi_{(g,\pi)}^F[\xi]=0$.

It is evident that the Hamiltonian (\ref{RTham}) has the form of a Lagrange function, where we seek to find extrema of $\xi^\mu_\infty\mathbb{P}_\mu$ subject to the constraints being satisfied. As such, we need to make use of the following Lagrange multipliers theorem for Banach manifolds (cf. Theorem 6.3 of \cite{phasespace}).
\begin{theorem}\label{banach}
Suppose $K:B_1\rightarrow B_2$ is a $C^1$ map between Banach manifolds, such that $DK_u:T_uB_1\rightarrow T_{K(u)}B_2$ is surjective, with closed kernel and closed complementary subspace for all $u\in K^{-1}(0)$. Let $f\in C^1(B_1)$ and fix $u\in K^{-1}(0)$, then the following statements are equivalent:
\begin{enumerate}[(i)]
\item For all $v\in\ker DK_u$, we have
\begin{equation}Df_u(v)=0.\end{equation}
\item There is $\lambda\in B_2^*$ such that for all $v\in B_1$,
\begin{equation}Df_u(v)=\left<\lambda,DK_u(v)\right>,\end{equation}	
where $\left< \, , \right>$ refers to the natural dual pairing.
\end{enumerate}
\end{theorem}

From this, we prove the following.
\begin{theorem}\label{thmEcrit}
Let $\xi_\infty\in\R^{1+3}$ be some fixed future-pointing timelike vector, $(s,S)\in L^1$, and define $E^{(\xi_\infty)}(g,\pi)\in C^\infty(\mathcal{C}(s,S))$ by
$$E^{(\xi_\infty)}(g,\pi)=\xi^\mu_\infty\mathbb{P}_\mu(g,\pi).$$
For $(g,\pi)\in\mathcal{C}(s,S)$, the following statements are equivalent:
\begin{enumerate}[(i)]
\item For all $(h,p)\in T_{(g,\pi)}\mathcal{C}(s,S)$,
$$DE^{(\xi_\infty)}_{(g,\pi)}[h,p]=0.$$
\item There exists $\xi\in\Xi$ satisfying
$$D\Phi^F_{(g,\pi)}[\xi]=0.$$
\end{enumerate}

\end{theorem}
\begin{proof}
Assume $(i)$ holds for some $(\tilde g,\tilde \pi)$; we first show $(i)\implies(ii)$. Let $K(g,\pi)=\Phi(g,\pi)-(s,S)$ and let $f(g,\pi)=\mathcal{H}^{(\xi)}(g,\pi)$ for some $\xi\in\Xi$, then condition $(i)$ of Theorem \ref{banach} is satisfied. It follows that there exists $\lambda\in L^2_{-5/2}$ such that
$$D\mathcal{H}^{(\xi)}_{(\tilde g,\tilde \pi)}[h,p]=\int_M \lambda\cdot D\Phi_{(\tilde g,\tilde \pi)}[h,p],$$
for all $(h,p)\in T_{(\tilde g,\tilde \pi)}\mathcal{F}$, which combined with (\ref{DHphi}), gives
$$-\int_M(h,p)\cdot D\Phi^F_{(\tilde g,\tilde \pi)}[\xi]=\int_M\lambda\cdot D\Phi_{(\tilde g,\tilde \pi)}[h,p].$$
Now $D\Phi^F_{(\tilde g,\tilde \pi)}[\xi] \in L^2_{-5/2}\times W^{1,2}_{-3/2}$, so Proposition \ref{propweakstrong} then implies 
$$D\Phi_{(\tilde g,\tilde \pi)}^F[\xi+\lambda]=0,$$
and $\lambda\in\Hdirh(M)$, which in turn implies $(\xi+\lambda)\in\Xi$.\\

Conversely, assuming $(ii)$ holds at some $(\tilde g,\tilde \pi)$, it follows from (\ref{DHphi}) that
$$D\mathcal{H}^{(\xi)}_{(\tilde g,\tilde \pi)}[h,p]=0,$$
for all $(h,p)\in T_{(\tilde g,\tilde \pi)}\mathcal{F}$. Then by the definition of $\mathcal{H}^{(\xi)}$, we have
$$D\mathcal{H}^{(\xi)}_{(\tilde g,\tilde \pi)}[h,p]=DE^{(\xi_\infty)}_{(\tilde g,\tilde \pi)}[h,p]=0,$$
for all $(h,p)\in\mathcal{C}(s,S)$; that is, $(i)$ holds.
\end{proof}
Physically, $E^{(\xi_\infty)}$ is interpreted as the total energy viewed by an observer at infinity, whose worldline is generated by $\xi_\infty$. So Theorem \ref{thmEcrit} may be interpreted as the statement that critical points of the energy measured by $\xi_\infty$, correspond to solutions with Killing vectors asymptotic to $\xi_\infty$.

Let $\eta$ be the Minkowski metric with signature $(-,+,+,+)$, and define $\mathbb{P}^\mu=\eta^{\mu\nu}\mathbb{P}_\nu$. Further define the total mass, $m=\frac{-\mathbb{P}^\mu\mathbb{P}_\mu}{\sqrt{|\mathbb{P}^\mu\mathbb{P}_\mu|}}$. Recall that we have not imposed conditions on the boundary mean curvature; that is, we include initial data for which the positive mass theorem fails. Away from $m=0$, this is a smooth function on $\mathcal{C}(s,S)$ when $(s,S)\in L^1$. With this in mind, we have the following corollary of Theorem \ref{thmEcrit}.
\begin{corollary}\label{cormcrit}
Suppose $(g,\pi)\in\mathcal{C}(s,S)$ with $(s,S)\in L^1$, and $\mathbb{P}^\mu$ is a past-pointing timelike vector, then the following statements are equivalent:
\begin{enumerate}[(i)]
\item For all $(h,p)\in T_{(g,\pi)}\mathcal{C}(s,S)$, $Dm_{(g,\pi)}[h,p]=0$.
\item $(g,\pi)$ is a stationary initial data set, whose stationary Killing vector is proportional to $\mathbb{P}$ at infinity and vanishes on $\Sigma$.
\end{enumerate}
\end{corollary}
It is worth noting that a Killing vector that is asymptotically constant, must in fact be proportional to $\mathbb P$ at infinity \cite{BeigChruscielKilling}.
\begin{proof}
We first show the implication $(i)\implies(ii)$. Let $\xi_\infty^\mu=-\frac1m\mathbb{P}^\mu$ be a future-pointing unit timelike vector, parallel to $\mathbb{P}^\mu$. It then follows that $E^{(\xi_\infty)}(g,\pi)=m$, so $(i)$ implies condition $(i)$ of Theorem \ref{thmEcrit} and $(ii)$ follows.

Conversely, if $(ii)$ holds, then possibly after rescaling, we have some $\xi\in\Xi$, where $\xi_\infty^\mu=-\frac1m\mathbb{P}^\mu$, satisfying $D\Phi^F_{(g,\pi)}[\xi]=0$. Again, $E^{(\xi_\infty)}(g,\pi)=m$ and Theorem \ref{thmEcrit} implies $(i)$.
\end{proof}
Provided $g$ is sufficiently smooth, the stationarity conclusion can in fact be replaced with staticity by the following argument. It is well-known that if a Killing vector field vanishes identically on a closed spacelike 2-surface, then that 2-surface is the bifurcation surface of a bifurcate Killing horizon (see, for example \cite{WaldQFT}). Furthermore, a result of Chru\'sciel and Wald \cite{chruscielwald} implies the existence of a maximal spacelike hypersurface in the full spacetime containing the bifurcation surface. Then a staticity theorem of Sudarsky and Wald can be applied \cite{SW1} (cf. Section 7 of \cite{chruscielcostauniqueness}), which states, under the assumption of the existence of a maximal spacelike hypersurface, if the stationary Killing vector generates the horizon, then the solution is static. That is, for the vacuum case, critical points of the mass occur exactly when the solution is the region exterior to a Schwarzschild black hole. It follows that for generic choices of $\og$ on $\Sigma$, there are no smooth critical points of the mass functional.

\begin{remark}
The same analysis can be performed with $\pi\equiv0$, considering only the Hamiltonian constraint. In this case, the mass and energy are interchangeable, and we only have the lapse as the Lagrange multiplier. The conclusion from the above analysis is then that critical points of the mass correspond to static solutions, as the Killing vector is necessarily hypersurface orthogonal (cf. Theorem 8 of \cite{Corvino2000}).
\end{remark}

\subsection{Geometric boundary data}
The asymptotic value of the stationary Killing vector field predicted by Theorem \ref{thmEcrit}, comes from our choice of $\xiref$, which above we chose to be supported away from $\Sigma$. However, if we allow $\xiref$ to be nonzero on $\Sigma$ then the energy-momentum can no longer be expressed as integrals over $M$, and expression (\ref{DHphi}) no longer holds. To deal with this, we leave $\xiref$ unchanged in the definition of $\mathbb{P}$ and we introduce $\xi_\Sigma=(\xi^0_\Sigma,0,0,0)$ with support near $\Sigma$ and $\xi^0_\Sigma$ constant on $\Sigma$. We then modify the Hamiltonian to allow for ${\xi\in\hat\Xi:=\{\xi:\xi-\xiref-\xi_\Sigma\in\Hdirh(M)\}}$
\begin{equation}\label{RTham2}
\hat{\mathcal{H}}^{(\xi)}(g,\pi)=16(\pi \xi^\mu_\infty\mathbb{P}_\mu-\xi^0_\Sigma m_{BY}(g;\Sigma))-\int_M\xi^\mu\Phi_\mu(g,\pi),
\end{equation}
where $m_{BY}$ is the Brown-York mass. This is given by
$$m_{BY}(\Sigma)=\frac1{8\pi}\oint_\Sigma (h_0-h_g)\,dS_g,$$
where $h$ is the mean curvature of $\Sigma$ in $M$ and $h_0$ is the mean curvature of $\Sigma$ isometrically embedded in $\R^3$, both computed with respect to the unit normal pointing towards infinty. Note that we could simply replace the term $-16\pi m_{BY}$ with twice the mean curvature of $\Sigma$ in $M$, as the addition of a constant does not change the equations of motion; however, this Hamiltonian is more intuitive as it gives a sensible measure of the energy of the system.
In coordinates adapted to $\Sigma$, the linearisation of $m_{BY}(g;\Sigma)$ is given by (cf. \cite{miao2003existence})
$$ 16\pi Dm^{BY}_{(g,\pi)}(\Sigma)[h,p]=\oint_\Sigma\left(\nabla_nh_{nn}-H_\Sigma(g)h_{nn}+2h^{AB}K_{AB}-2\nabla^ih_{in}+\nabla_n\tr_gh\right)dS,$$
where $A,B=1,2$ are coordinates on $\Sigma$, $n$ is the normal direction, and $K_{AB}$ is the second fundamental form of $\Sigma$ with respect to $g$ and $n$. Since $h$ vanishes on $\Sigma$, this reduces to
$$ 16\pi Dm^{BY}_{(g,\pi)}(\Sigma)[h,p]=\oint_\Sigma\left(\nabla_n\tr_gh-\nabla_nh_{nn}\right)dS,$$
where we have also made use of the divergence theorem. Now it is straightforward to check (cf. \ref{boundaryterms}) we have
$$(h,p)\cdot D\Phi^F_{(g,\pi)}[\xi]-\xi\cdot D\Phi_{(g,\pi)}[h,p]=16\pi\xi^0_\Sigma Dm^{BY}_{(g,\pi)}(\Sigma)[h,p].$$
This then gives us (cf. \ref{DHphi})
\begin{equation}\label{DHphi2}
D\hat{\mathcal{H}}^{(\xi)}_{(g,\pi)}[h,p]=-\int_M(h,p)\cdot D\Phi^F_{(g,\pi)}[\xi],
\end{equation}
for all $(h,p)\in T_{(g,\pi)}\mathcal{F}$. At this point, it only requires superficial modifications to the proofs of Theorem \ref{thmEcrit} and Corollary \ref{cormcrit}, to obtain the following.

\begin{theorem}\label{thmEcrit2}
Let $\xi_\infty\in\R^{1+3}$ be some fixed future-pointing timelike vector, $\xi_\Sigma\in\R$ be some fixed constant, $(s,S)\in L^1$, and define $\hat{E}^{(\xiref)}(g,\pi)\in C^\infty(\mathcal{C}(s,S))$ by
$$\hat{E}^{(\xiref)}(g,\pi)=\xi^\mu_\infty\mathbb{P}_\mu(g,\pi)-\xi_\Sigma m^{BY}(g;\Sigma).$$
For $(g,\pi)\in\mathcal{C}(s,S)$, the following statements are equivalent:
\begin{enumerate}[(i)]
\item For all $(h,p)\in T_{(g,\pi)}\mathcal{C}(s,S)$,
$$D\hat{E}^{(\xiref)}_{(g,\pi)}[h,p]=0.$$
\item There exists $\xi\in\hat\Xi$ satisfying
$$D\Phi^F_{(g,\pi)}[\xi]=0.$$
\end{enumerate}

\end{theorem}
Note that this version of the theorem does not force the Killing vector to vanish on the boundary, but rather it is orthogonal to the initial data hypersurface there. By fixing $\xi_\Sigma=1$ on $\Sigma$ Corollary \ref{cormcrit} becomes:
\begin{corollary}\label{cormcrit2}
Suppose $(g,\pi)\in\mathcal{C}(s,S)$ with $(s,S)\in L^1$, and $\mathbb{P}^\mu$ is a past-pointing timelike vector, then the following statements are equivalent:
\begin{enumerate}[(i)]
\item For all $(h,p)\in T_{(g,\pi)}\mathcal{C}(s,S)$, $Dm_{(g,\pi)}[h,p]=Dm^{BY}_{(g,\pi)}(\Sigma)[h,p]$.
\item $(g,\pi)$ is a stationary initial data set, whose stationary Killing vector is proportional to $\mathbb{P}^\mu$ at infinity and $(-m_0,0,0,0)$ on $\Sigma$, with the same constant of proportionality.
\end{enumerate}
\end{corollary}
\begin{remark}
Corollary \ref{cormcrit2} is not quite what is required to demonstrate that critical points of the mass, over the space of extensions satisfying Bartnik's geometric boundary conditions, are stationary/static; however, it does illuminate the connection. To prove such a result, we would likely need to conduct the entire analysis again, replacing $\Phi$ with $(\Phi,H)$, so as to include the boundary conditions in the constraint map itself. However, this would likely require significant changes to the analysis presented here.
\end{remark}

By choosing different conditions on $\xi$, both at infinity and on $\Sigma$, we will obtain different conditions for solutions to be stationary; essentially, these ideas can be used to find the appropriate condition for the existence of a Killing vector with prescribed boundary conditions. In \cite{PhysRevD.90.104034}, we use similar ideas to prove that the first law of black hole mechanics gives a condition for stationarity, when the boundary conditions on the Killing vector are inspired by bifurcate Killing horizons. Here we can include the quasilocal generalised angular momentum used in \cite{PhysRevD.90.104034} to obtain a similar result, sans the area/surface gravity term (as the metric is fixed on $\Sigma$ here). One can also infer that $\hat{E}^{(\xiref)}$ has no critical points when $\xi_\infty=0$ from the fact that $D\Phi^F$ has trivial kernel in $L^2_{-1/2}$. That is, one immediately has the expected, or perhaps even obvious, result that the Brown-York mass (equivalently, the mean curvature of $\Sigma$) has no critical points.

\section{Acknowledgements}
This research was supported in part by a UNE Research Seed Grant. The author would like to thank the Institut Henri Poincar\'e, where part of this research was completed.

\bibliographystyle{abbrv}
\bibliography{../../refsnew}

\begin{thebibliography}{10}

\bibitem{AndersonKhuri}
M.~T. Anderson and M.~A. Khuri.
\newblock On the {B}artnik extension problem for the static vacuum {E}instein
  equations.
\newblock {\em Class. Quant. Grav.}, 30(12):125005, 2013.

\bibitem{AF}
R.~Bartnik.
\newblock The mass of an asymptotically flat manifold.
\newblock {\em Comm. Pure. Appl. Math.}, 39:661--693, 1986.

\bibitem{qlm}
R.~Bartnik.
\newblock New definition of quasilocal mass.
\newblock {\em Phys. Rev. Lett.}, 62(20):845--885, 1989.

\bibitem{tsinghua}
R.~Bartnik.
\newblock Energy in general relativity.
\newblock {\em Tsing Hua Lectures on Geometry and Analysis,(Hsinchu,
  1990-1991)}, pages 5--27, 1997.

\bibitem{Bartnik02}
R.~Bartnik.
\newblock Mass and 3-metrics of nonnegative scalar curvature.
\newblock In {\em Proc. Int. Cong. Math., Vol III}, pages 231--240, 2002.

\bibitem{phasespace}
R.~Bartnik.
\newblock Phase space for the {E}instein equations.
\newblock {\em Comm. Anal. Geom.}, 13(5):845--885, 2005.

\bibitem{BeigChruscielKilling}
R.~Beig and P.~Chru\'sciel.
\newblock {K}illing vectors in asymptotically flat space--times. {I}.
  {A}symptotically translational {K}illing vectors and the rigid positive
  energy theorem.
\newblock {\em J. Math. Phys}, 37(4):4, 1996.

\bibitem{ellipticsys}
Y.~Choquet-Bruhat and D.~Christodoulou.
\newblock Elliptic systems in spaces on {$H_{s,\delta}$} manifolds which are
  {E}uclidean at infinity.
\newblock {\em Acta Mathematica}, 146(1):129--150, 1981.

\bibitem{chruscielcostauniqueness}
P.~T. Chru\'sciel and J.~L. Costa.
\newblock On uniqueness of stationary vacuum black holes.
\newblock {\em arXiv:0806.0016}, 2008.

\bibitem{chruscielwald}
P.~T. Chru\'sciel and R.~M. Wald.
\newblock Maximal hypersurfaces in stationary asymptotically flat spacetimes.
\newblock {\em Comm. Math. Phys.}, 163(3):561--604, 1994.

\bibitem{Corvino2000}
J.~Corvino.
\newblock Scalar curvature deformation and a gluing construction for the
  {E}instein constraint equations.
\newblock {\em Comm. Math. Phys.}, 214(1):137--189, 2000.

\bibitem{boundedl2}
S.~Klainerman, I.~Rodnianski, and J.~Szeftel.
\newblock The bounded {$L^2$} curvature conjecture.
\newblock {\em arXiv:1204.1767 [To appear in Inventiones Mathematicae]}, 2015.

\bibitem{MantoulidisSchoen}
C.~Mantoulidis and R.~Schoen.
\newblock On the {B}artnik mass of apparent horizons.
\newblock {\em Class. Quant. Grav.}, 32(20):205002, 2015.

\bibitem{maxwell2005solutions}
D.~Maxwell.
\newblock Solutions of the {E}instein constraint equations with apparent
  horizon boundaries.
\newblock {\em Communications in mathematical physics}, 253(3):561--583, 2005.

\bibitem{PhysRevD.90.104034}
S.~McCormick.
\newblock First law of black hole mechanics as a condition for stationarity.
\newblock {\em Phys. Rev. D}, 90:104034, Nov 2014.

\bibitem{mythesis}
S.~McCormick.
\newblock {\em The phase space for the {E}instein-{Y}ang-{M}ills equations,
  black hole mechanics, and a condition for stationarity}.
\newblock PhD thesis, Monash University, 2014.

\bibitem{massmini}
S.~McCormick.
\newblock A note on mass-minimising extensions.
\newblock {\em Gen. Rel. Grav.}, 47(12), 2015.

\bibitem{McOwen}
R.~C. McOwen.
\newblock The behavior of the {L}aplacian on weighted {S}obolev spaces.
\newblock {\em Comm. Pure and Applied Math.}, 32(6):783--795, 1979.

\bibitem{miao2003existence}
P.~Miao.
\newblock On existence of static metric extensions in general relativity.
\newblock {\em Communications in mathematical physics}, 241(1):27--46, 2003.

\bibitem{pengzi2003}
P.~Miao.
\newblock Variational effect of boundary mean curvature on {ADM} mass in
  general relativity.
\newblock {\em arXiv preprint math-ph/0309045}, 2003.

\bibitem{miao2004positive}
P.~Miao.
\newblock Positive mass theorem on manifolds admitting corners along a
  hypersurface.
\newblock {\em Advances on Theoretical and Mathematical Physics}, 6:1163--1182,
  2004.

\bibitem{moncrief1}
V.~Moncrief.
\newblock Spacetime symmetries and linearization stability of the {E}instein
  equations. {I}.
\newblock {\em J. Math. Phys.}, 16(3):493--498, 1975.

\bibitem{NW}
L.~Nirenberg and H.~F. Walker.
\newblock The null spaces of elliptic partial differential operators in
  {$\mathbb{R}^n$}.
\newblock {\em J. Math. Analysis and Applications}, 42(2):271--301, 1973.

\bibitem{RT}
T.~Regge and C.~Teitelboim.
\newblock The role of surface integrals in general relativity.
\newblock {\em Ann. Phys.}, (88):286--318, 1974.

\bibitem{SW1}
D.~Sudarsky and R.~M. Wald.
\newblock Extrema of mass, stationarity, and staticity, and solutions to the
  {E}instein-{Y}ang-{M}ills equations.
\newblock {\em Phys. Rev. D}, 46(4):1453--1474, 1992.

\bibitem{QLMreview}
L.~B. Szabados.
\newblock Quasi-local energy-momentum and angular momentum in general
  relativity: {A} review article.
\newblock {\em Living Rev. Relativity}, 7(4), 2004.

\bibitem{WaldQFT}
R.~M. Wald.
\newblock {\em Quantum field theory in curved spacetime and black hole
  thermodynamics}.
\newblock University of Chicago Press, 1994.

\end{thebibliography}
\end{document}